%% file: article.tex
\documentclass{article}
\usepackage{amsfonts}
\usepackage{amsthm}
\usepackage{amsmath}
\usepackage{algorithm}
\usepackage{algorithmic}

\newtheorem{theorem}{Theorem}
\newtheorem{proposition}{Proposition}
\newtheorem{corollary}{Corollary}
\newtheorem{lemma}{Lemma}
\newtheorem{remark}{Remark}
\newtheorem{definition}{Definition}

\newcommand{\C}{\mathcal{C}} 

\newcommand{\F}{\mathbb{F}} 

\newcommand{\poly}[2]{#1[X]_{< #2}}

\newcommand{\floor}[1]{\left\lfloor {#1} \right\rfloor}

\newcommand{\inner}[2]{\left\langle #1,#2 \right\rangle}

\newcommand{\myproof}{proof}

\newcommand{\eqdef}{:=}

\newcommand{\matfqs}{M_{\ell\times\ell}(\F_{q^s})}
\newcommand{\LRS}{\mathcal{L}}
\newcommand{\RRS}{\mathcal{R}}

\DeclareMathOperator{\pr}{pr}

\DeclareMathOperator{\qbch}{Q-BCH}

\DeclareMathOperator{\GL}{GL}

\DeclareMathOperator{\blockw}{Block-w}

\addtocounter{MaxMatrixCols}{10}

\title{On the decoding of quasi-BCH codes}
\author{Morgan Barbier\thanks{ 
  Universit\'{e} de Caen - GREYC,
  Boulevard Mar\'{e}chal Juin,
  BP 5186, 14032 Caen}\\
  \texttt{morgan.barbier@unicaen.fr}
  \and 
  Cl\'{e}ment Pernet\thanks{
  Univ. Joseph Fourier
  INRIA/LIG-MOAIS,
  51 avenue Jean Kuntzmann,
  38330 Montbonnot}\\
  \texttt{clement.pernet@imag.fr}
  \and
  Guillaume Quintin\thanks{
  \'{E}cole polytechnique -
  LIX, 
  91128 Palaiseau Cedex}\\
  \texttt{quintin@lix.polytechnique.fr}}
\date{2012/12/21}

\begin{document}
\sloppypar
\maketitle
\begin{abstract}
  \input{abstract.tex}
\end{abstract}
\paragraph{keywords:} Quasi-cyclic code, quasi-BCH code, BCH code,
Reed-Solomon, interleaved code
\section{Introduction}
\input{blabla.tex}
\subsection{Our contributions}
\input{contribution.tex}
\subsection{Related work}
\input{related_work.tex}
\section{Prerequisites}
\subsection{Reed-Solomon codes over rings}
\input{rs.tex}
\subsection{Quasi cyclic and quasi BCH codes}
\input{quasi_bch.tex}
\section{Reed-Solomon codes and quasi-BCH codes}
\subsection{The relation between quasi-BCH and Reed-Solomon codes}
\input{duality.tex}
\subsection{The Welch-Berlekamp algorithm for quasi-BCH codes}
\input{decoding1.tex}
\section{Quasi-BCH codes as interleaved codes}
\input{interleaved.tex}
\bibliographystyle{alpha}
\bibliography{biblio}
\end{document}

%% file: abstract.tex
In this paper we investigate the structure of quasi-BCH
codes.
In the first part of this paper we show that quasi-BCH codes
can be derived from
Reed-Solomon codes over square matrices extending the known
relation about classical BCH and Reed-Solomon codes. This
allows us to adapt the Welch-Berlekamp algorithm to quasi-BCH
codes.
In the second part  of this paper we show that quasi-BCH codes
can be seen as subcodes of interleaved Reed-Solomon codes over
finite fields.
This provides another approach for decoding quasi-BCH
codes.

%% file: blabla.tex
Many codes with best known minimum distances are quasi-cyclic codes
or derived from them \cite{LinSol2003,codetables}. This family of
codes is therefore very interesting. Quasi-cyclic codes were 
studied and applied in the context of McEliece's cryptosystem
\cite{McEli78,BerCayGabOtm2009} and Niederreiter's
\cite{Nied86,LiDengWang94}.
They permit to reduce the size of keys in opposition to Goppa
codes. However, since the decoding of random quasi-cyclic codes is
difficult, only quasi-cyclic alternant codes were proposed for the
latter cryptosystem. The high structure of alternant codes is
actually a weakness and two cryptanalysis were proposed in
\cite{FauOtmPerTil2010,UmaLea2010}

%% file: contribution.tex
In this paper we investigate the structure of quasi-BCH
codes.
In the first part of this paper we show that quasi-BCH codes
can be derived from Reed-Solomon codes over square matrices.
It is well known that BCH codes can be obtained from
Reed-Solomon codes \cite[Theorem~2, page 300]{SloMacWil86}.
We extend this property to quasi-BCH codes which allows us
to adapt the Welch-Berlekamp algorithm to quasi-BCH codes.
\begin{theorem}
  Let $\Gamma \in \matfqs$ be a primitive
  $m$-th root of unity and $\C = \qbch_q(m,\ell,\delta,\Gamma)$.
  Then there exists a RRS code $\RRS$ over the ring
  $\matfqs$ with parameters
  $\left[ n,n - \delta + 1 \right]_{\matfqs}$
  and a $\F_q$-linear, $F_q$-isometric embedding
  $\psi : \C \rightarrow \RRS$.
\end{theorem}
In the second part we show that quasi-BCH codes can be seen
as subcodes of interleaved Reed-Solomon codes.
\begin{theorem}
  The quasi-BCH code $\C$ over $\F_q$ is an interleaved code
  of $\ell$ subcodes of Reed-Solomon codes over $\F_{q^{s'}}$
  in the following sense:
  there exists $\ell$ Reed-Solomon codes $\C_1,\dots,\C_{\ell}$ over
  $\F_q$ and an isometric isomorphism from $\C$, equipped with the
  $\ell$-block distance, to a subcode of the interleaved code with
  respect to $\C_1,\dots,\C_{\ell}$.
\end{theorem}

%% file: related_work.tex
In \cite{LalFitz2001,LinSol2001}, $\ell$-quasi-cyclic codes of length
$m\ell$ are seen as $R$-submodules of $R^\ell$ for a certain ring
$R$. However, in \cite{LalFitz2001}, Gr\"{o}bner bases are used in order
to describe polynomial generators of quasi-cyclic codes whereas
in \cite{LinSol2001}, the authors decompose quasi-cyclic codes as
direct sums of shorter linear codes over various extensions of $\F_q$
(when $\gcd(m,q)=1$). This last work leads to an interesting  trace
representation of quasi-cyclic codes. In \cite{CayChaAbd2010}, the
approach is more analogous to the cyclic
case. The authors consider the factorization of
$X^m - 1 \in M_\ell(F_q)[X]$ with reversible polynomials in order to
construct $\ell$-quasi-cyclic codes canceled by those polynomials and
called $\Omega(P)$-codes. This leads to the
construction of self-dual codes and codes beating known bounds.
But the factorization of univariate polynomials over a matrix ring
remains difficult. In \cite{Chabot2011} the author gives an
improved method for particular cases of the latter factorization
problem.

%% file: rs.tex
We recall some basic definitions of Reed-Solomon codes over rings in
this section. We let $A$ be a ring with identity, we denote by
$A^{\times}$ the \emph{group of units} of $A$ and by $Z(A)$ the
\emph{center} of $A$, the commutative subring of $A$ consisting of
all the elements of $A$ which commutes with all the other elements
of $A$. We denote by $A[X]$ the ring of polynomials over $A$ and
by $\poly{A}{k}$ the polynomials over $A$ of degree at most $k - 1$.
\begin{definition}
\label{def:eval}
  Let
  \begin{equation*}
    f = \sum_{i = 0}^d f_i X^i \in A[X]
  \end{equation*}
  be a polynomial with coefficients in $A$ and
  $a \in A$. We call \emph{left evaluation of $f$ at $a$} the
  quantity
  \begin{equation*}
    f(a) \eqdef \sum_{i = 0}^d f_i a^i \in A
  \end{equation*}
  and \emph{right evaluation of $f$ at $a$} the quantity
  \begin{equation*}
    (a)f \eqdef \sum_{i = 0}^d a^i f_i \in A.
  \end{equation*}
\end{definition}
\begin{remark}
  For $f,g \in A[X]$ and $a \in A$, we obviously have
  $f(a) = (a)f$ whenever $a \in Z(A)$,
  $(f + g)(a) = f(a) + g(a)$,
  $(a)(f + g) = (a)f + (a)g$. If $a$ commutes with all
  the coefficients of $g$ we also have
  $(fg)(a) = f(a)g(a)$ and $(a)(gf) = (a)g(a)f$.
\end{remark}
\begin{definition}
  Let $0 < k \leq n$ be two integers.
  Let $(x_1,\dots,x_n)$ and $v = (v_1,\dots,v_n)$
  be two vectors of $A^n$ be such that
  $x_i - x_j \in A^{\times}$ and $x_i x_j = x_j x_i$
  for all $i \neq j$ and $v_i \in A^{\times}$ for all $i$.

  The \emph{left} submodule of $A^n$ generated by the vectors
  \begin{equation*}
    (f(x_1) \cdot v_1,\dots,f(x_n) \cdot v_n) \in A^n
    \text{ with }
    f \in \poly{A}{k}
  \end{equation*}
  is called a \emph{left generalized Reed-Solomon code (LGRS) over
  $A$ with parameters $[v,x,k]_A$} or $[n,k]$ if there is no confusion
  on $x$ and $v$.

  The \emph{right} submodule of $A^n$ generated by the vectors
  \begin{equation*}
    (v_1 \cdot (x_1)f,\dots,v_n \cdot (x_n)f) \in A^n
    \text{ with }
    f \in \poly{A}{k}
  \end{equation*}
  is called a \emph{right generalized Reed-Solomon code (RGRS) over
  $A$ with parameters $[v,x,k]_A$} or $[n,k]$ if there is no confusion
  on $x$ and $v$. The vector $x$ is called the \emph{support} of the
  code. If $v = (1,\dots,1)$, the codes constructed above are called
  left Reed-Solomon (LRS) and right Reed-Solomon (RRS) codes.
\end{definition}
\begin{definition}
  Let $x = (x_1,\dots,x_n) \in A^n$. We call the \emph{Hamming
  weight of $x$} the number of nonzero coordinates.
  \begin{equation*}
    w(x) \eqdef w(x_1,\dots,x_n) =
    \left| \{ i : x_i \neq 0 \} \right|.
  \end{equation*}
  Let $y = (y_1,\dots,y_n) \in A^n$.
  The \emph{Hamming distance between $x$ and $y$} is
  \begin{equation*}
    d(x,y) = w(x - y) = \left| \{ i : x_i \neq x_j \} \right|.
  \end{equation*}
  The \emph{minimum distance} of any subset $S \subseteq A^n$ is
  defined as
  \begin{equation*}
    \min \left\{ d(x,y) : x,y \in S \text{ and } x \neq y \right\}.
  \end{equation*}
\end{definition}
\begin{proposition}
\label{prop:rs}
  A LGRS (resp. RGRS) code is a free left (resp. right) submodule
  of $A^n$. A LGRS (resp. RGRS) code with parameters $[n,k]$ has
  minimum distance $n - k + 1$.
\end{proposition}
\begin{proof}
  It suffices to see that the maps
  \begin{equation*}
    \begin{array}{rcl}
      A^n             & \longrightarrow & A^n \\
      (a_1,\dots,a_n) & \longmapsto     & (a_1 v_1,\dots,a_n v_n) \\
      (a_1,\dots,a_n) & \longmapsto     & (v_1 a_1,\dots,v_n a_n)
    \end{array}
  \end{equation*}
  are respectively left and right isometric automorphisms of $A^n$.
\end{proof}

%% file: quasi_bch.tex
Quasi cyclic codes form an important family of codes defined
as follow.
\begin{definition}
  Let $T:\F_q^n \rightarrow \F_q^n$ to be the left cyclic shift
  defined by
  \begin{equation*}
    T(c_1, c_2, \dots, c_n) = (c_2,c_3,\dots,c_1).
  \end{equation*}
  We call
  \emph{$\ell$-quasi-cyclic code over $\F_q$ of length $n$}
  any code of length $n$ over $\F_q$ stable by $T^{\ell}$.
  If the context is clear we will
  simply say \emph{$\ell$-quasi-cyclic code}.
\end{definition}
We will focus in this paper on quasi-BCH codes which form
a subfamily of quasi-cyclic codes. They can be seen as a
generalization of BCH codes in the context of quasi-cyclic
codes. For we need primitive roots of unity defined in
a extension of $\F_q$, say $\F_{q^s}$ to construct BCH
codes over $\F_q$.
\begin{proposition}
\label{prop:prim_root}
  Then there exists a primitive $q^{s\ell} - 1$-th root of unity in
  $M_{\ell}(\F_{q^s})$.
\end{proposition}
\begin{proof}
  The proof can be found in
  \cite[Proposition~16, page~911]{BCQ2011}.
\end{proof}
\begin{definition}
\label{defi:qbch}
  Let $\Gamma$ be a primitive $m$-th root of unity in
  $M_{\ell}(\F_{q^s})$ and $\delta \leq m$.
  We define the $\ell$-quasi-BCH code of length $m\ell$,
  with respect to $\Gamma$, with designed minimum distance
  $\delta$, over $\F_q$ by
  \begin{multline*}
    \qbch_q(m,\ell,\delta,\Gamma) \eqdef \\
    \left\lbrace
      (c_1,\ldots,c_m) \in (\F_q^\ell)^m :
        \sum_{j = 0}^{m - 1} (\Gamma^i)^j (c_{j+1})^T = 0
        \text{ for } i = 1,\ldots,\delta - 1
    \right\rbrace.
  \end{multline*}
  Note that $\qbch_q(m,\ell,\delta,\Gamma)$ is a quasi-cyclic code.
\end{definition}
\begin{definition}
  The \emph{$\ell$-block weight} of
  $(x_{11},\dots,x_{1\ell},\dots,x_{m1},\dots,x_{m\ell})
   \in \F_q^{m\ell}$ is defined to be
  \begin{equation*}
    \blockw_{\ell}(x) \eqdef
    \left|
      \left\{ i : (x_{i1},\dots,x_{i\ell}) \neq 0 \right\}
    \right|.
  \end{equation*}
  The \emph{$\ell$-block distance} between $x,y \in \F_q^{m\ell}$
  is defined to be $\blockw_{\ell}(x - y)$.
\end{definition}

%% file: duality.tex
We show in this section that under certain assumptions on the
support of Reed-Solomon codes, the dual of a LRS code is a
RRS code. From this fact we show that quasi-BCH can be constructed
from Reed-Solomon codes over square matrices rings. In this
Subsection we let $A$ designate a finite ring with identity.
\begin{definition}
  Let $x = (x_1,\dots,x_n)$ and $y = (y_1,\dots,y_n)$ be two vectors
  of $A^n$. The \emph{inner product} is defined as
  \begin{equation*}
    \inner{x}{y} \eqdef \sum_{i = 0}^n x_i y_i.
  \end{equation*}
\end{definition}
\begin{remark}
\label{rem:inner}
  Let $S$ be a subset of $A^n$. Then the set
  $\{ x \in A^n : \forall s \in S, \inner{s}{x} = 0 \}$
  denoted by $S^{\perp}$
  is called the \emph{right dual of $S$} and
  is a right submodule of $A^n$. Similarly,
  Let $S$ be a subset of $A^n$. Then the set
  $\{ x \in A^n : \forall s \in S, \inner{x}{s} = 0 \}$
  denoted by ${}^{\perp}S$
  is called the \emph{left dual of $S$} and
  is a left submodule of $A^n$.
  Note that for all $x,y \in A^n$ and $\mu \in A$ we have
  $\mu \inner{x}{y} = \inner{\mu x}{y}$ and
  $\inner{x}{y} \mu = \inner{x}{y \mu}$.
\end{remark}
\begin{definition}
  We say that $a \in A$ is a \emph{primitive $m$-th root of unity} if
  $a^m = 1$ and $\forall 0 \leq i < m, (a^i - 1) \in A^{\times}$.
\end{definition}
\begin{remark}
  Let $x = (1,\gamma,\gamma^2,\dots,\gamma^{m - 1}) \in A^m$ where
  $\gamma$ is a primitive $m$-th root of unity. Then a RRS or LRS
  code whose support is $x$ is cyclic.
\end{remark}
\begin{proposition}
\label{prop:dual}
  Let $\gamma \in A$ be a primitive $m$-th root of unity.
  Let $x = (1,\gamma,\gamma^2,\dots,\gamma^{m - 1}) \in A^n$.
  Then the right (resp. left) dual of the LGRS (resp. RGRS) code with
  parameters
  $\left[ x,x,k \right]_A$ is the RRS (resp. LRS)
  code with parameters $[x,n - k]_A$.
\end{proposition}
\begin{proof}
  We denote respectively by $\LRS$ and $\RRS$ the left
  generalized Reed-Solomon code with parameters
  $[x,x,k]_A$ and the right Reed-Solomon code
  with parameters $[x,n - k]_A$.

  First note that $\LRS$ is generated by the vectors
  \begin{equation*}
    (1,\gamma^i,\gamma^{2i},\dots,\gamma^{(m - 1)i})
    \text{ for }
    i = 1,\dots,k
  \end{equation*}
  and that $\RRS$ is generated by the vectors
  \begin{equation*}
    (1,\gamma^i,\gamma^{2i},\dots,\gamma^{(m - 1)i})
    \text{ for }
    i = 0,\dots,n - k - 1.
  \end{equation*}
  And we have for $0 \leq i + j < n - 1$ in the commutative ring
  $Z(A)[\gamma]$
  \begin{equation*}
    \sum_{i = 0}^{m - 1} \gamma^{(i + 1)\ell} \cdot \gamma^{j\ell} =
    \sum_{i = 0}^{m - 1} \left( \gamma^{i + j + 1} \right)^{\ell} =
    \frac{1 - \left( \gamma^{i + j + 1}\right) ^ m}
         {1 - \gamma^{i + j + 1}} = 0.
  \end{equation*}
  Therefore, by Proposition~\ref{prop:rs} and Remark~\ref{rem:inner},
  $\LRS^{\perp} \subseteq \RRS$ and ${}^{\perp}\RRS \subseteq \LRS$.
  
  Again by Proposition~\ref{prop:rs} and Remark~\ref{rem:inner}
  an element $x \in A^n$ lies in $\LRS^{\perp}$ if and only if
  \begin{equation}
  \label{equ:dual}
    \left[
    \begin{pmatrix}
      1 & 1              & 1                 & \dots & 1 \\
      1 & \gamma         & \gamma^2          & \dots & \gamma^{m - 1} \\
      1 & \vdots         & \vdots            &       & \vdots \\
      1 & \gamma^{k - 1} & \gamma^{2(k - 1)} & \dots &
                                              \gamma^{(k - 1)(m - 1)}
    \end{pmatrix}
    \begin{pmatrix}
      1 &        &        & \\
        & \gamma &        & \\
        &        & \ddots & \\
        &        &        & \gamma^{m - 1}
    \end{pmatrix}
    \right]
    \begin{pmatrix}
      x_1 \\ x_2 \\ \vdots \\ x_n
    \end{pmatrix} = 0.
  \end{equation}
  But in the commutative ring $Z(A)[\gamma]$ the matrix
  \begin{equation*}
    H =
    \begin{pmatrix}
      1 & 1              & 1                 & \dots & 1 \\
      1 & \gamma         & \gamma^2          & \dots &
                                        \gamma^{2(k - 1)} \\
      1 & \vdots         & \vdots            &       & \vdots \\
      1 & \gamma^{k - 1} & \gamma^{2(k - 1)} & \dots &
                                      \gamma^{(k - 1)(k - 1)}
    \end{pmatrix}
    \in M_{k \times k} \left( Z(A)[\gamma] \right)
  \end{equation*}
  is invertible. Therefore $H$ is also invertible in
  $M_{k \times k}(A)$ and thus induces a group automorphism
  of $A^k$. If we let $x_H = (x_1,\dots,x_k)$,
  $x_U = (x_{k + 1},\dots,x_n)$, we can rewrite
  equation~\eqref{equ:dual} as
  \begin{equation*}
    \left(\begin{array}{c|c} H & U \end{array}\right)
    \left(\begin{array}{c} x_H \\ \hline x_U \end{array}\right) = 0
    \text{ and }
    \left(\begin{array}{c|c} H & 0 \end{array}\right)
    \left(\begin{array}{c} x_H \\ \hline 0 \end{array}\right) = -
    \left(\begin{array}{c|c} 0 & U \end{array}\right)
    \left(\begin{array}{c} 0 \\ \hline x_U \end{array}\right).
  \end{equation*}
  For each choice of $x_U$ we have only one possible value for $x_H$.
  Thus $|\LRS^{\perp}| = |A|^{n - k} = |\RRS|$
  by Proposition~\ref{prop:rs} and therefore
  $\LRS^{\perp} = \RRS$. Similarly, we have ${}^{\perp} \RRS = \LRS$.
\end{proof}
\begin{theorem}
\label{thm:qbch_rs}
  Let $\Gamma \in \matfqs$ be a primitive
  $m$-th root of unity and $\C = \qbch_q(m,\ell,\delta,\Gamma)$.
  Then there exists a RRS code $\RRS$ over the ring
  $\matfqs$ with parameters
  $\left[ n,n - \delta + 1 \right]_{\matfqs}$
  and a $\F_q$-linear, $F_q$-isometric embedding
  $\psi : \C \rightarrow \RRS$.
\end{theorem}
\begin{proof}
  A parity-check matrix of $\C$ is
  \begin{equation*}
    H = \begin{pmatrix}
      I_{\ell} & \Gamma         & \cdots & \Gamma^{m - 1} \\
      I_{\ell} & \Gamma^2       & \cdots & \Gamma^{2(m - 1)} \\
      \vdots   & \vdots         &        & \vdots \\
      I_{\ell} & \Gamma^{\delta - 1} & \cdots 
                                       & \Gamma^{(\delta - 1)(m - 1)}
    \end{pmatrix} \in M_{(\delta - 1)\ell,m\ell}(\F_{q^s}).
  \end{equation*}
  Remark that $H$ is a generator matrix of the LGRS code with
  parameters $[x,x,\delta - 1]_{\matfqs}$
  over the ring $\matfqs$ and by
  Proposition~\ref{prop:dual} its dual is the RRS with
  parameters $[x,\delta - 1]_{\matfqs}$.

  Now let
  \begin{equation*}
    \begin{array}{rcl}
      \psi : \C & \longrightarrow &
          \left( \matfqs \right)^m \\
      (c_{11},\dots,c_{1\ell},\dots,
       c_{m1},\dots,c_{m\ell}) & \longmapsto &
      \left[
        \begin{pmatrix}
          c_{11}    & 0      & \dots & 0 \\
          \vdots    & \vdots &       & \vdots \\
          c_{1\ell} & 0      & \dots & 0 \\
        \end{pmatrix},\dots,
        \begin{pmatrix}
          c_{m1}    & 0      & \dots & 0 \\
          \vdots    & \vdots &       & \vdots \\
          c_{m\ell} & 0      & \dots & 0 \\
        \end{pmatrix}
      \right]
    \end{array}.
  \end{equation*}
  Obviously, $\psi$ is $\F_q$-linear, injective and isometric
  and by the above remark we have $\psi(\C) \subseteq \RRS$.
\end{proof}
Theorem~\ref{thm:qbch_rs} generalizes the well-known
\cite[Theorem~2, page 300]{SloMacWil86}
relation between BCH codes and Reed-Solomon codes.
The above relation will allow us to adapt the unique decoding
algorithm from \cite{BCQ2012} to quasi-BCH codes.

%% file: decoding1.tex
In this
Subsection we let $A$ designate a finite ring with identity.
Before giving the Welch-Berlekamp decoding algorithm, we need to define
what the \emph{evaluation} of a bivariate polynomial over $A$ is. Let
$Q = \sum Q_{i,j} X^i Y^j \in A[X,Y]$ be such a polynomial. We define
the \emph{evaluation of $Q$ at $(a,b) \in A^2$} to be
\begin{equation*}
  (a,b)Q = \sum a^i b^j Q_{i,j} \in A.
\end{equation*}
Be careful of the order of $a$, $b$ and $Q_{i,j}$.
This choice will be explained
in the proof of Lemma~\ref{lem:bivariate_degree1_ev}.
Let $f \in A[X]$, we define \emph{the evaluation of $Q$ at $f$} to be
\begin{equation*}
  (X,f(X))Q = \sum X^j (f(X))^j Q_{i,j} \in A[X].
\end{equation*}
As in the univariate case, the evaluation maps defined above are not ring
homomorphisms in general.
\begin{lemma}
\label{lem:bivariate_degree1_ev}
  Let $g \in A[X]$, $Q \in A[X,Y]$ of degree at most $1$ in $Y$ and
  $a \in A$. Then
  \begin{equation*}
    (a)((X,g(X))Q) = (a,(a)g)Q.
  \end{equation*}
\end{lemma}
\begin{\myproof}
  We write
  \begin{align*}
    Q(X,Y) &= Q_0(X) + Q_1(X) Y \\
           &= Q_0(X) + \left( \sum_i Q_{1i} X^i \right) Y.
  \end{align*}
  The proof is an easy calculation:
  \begin{align*}
    (a)((X,g(X))Q) &= (a) \left( Q_0(X)
                    + \sum_{i} X^i g(X) Q_{1i} \right) \\
                   &= (a)Q_0 + \sum_{i} a^i (a)g Q_{1i} \\ 
                   &= (a,(a)g)Q \text{ by definition}.
  \end{align*}
\end{\myproof}
We let $\C = \qbch_q(m,\ell,\delta,\Gamma)$,
$\tau = \floor{\frac{\delta - 1}{2}}$, $n = m$,
$k = n - \delta + 1$ and
\begin{equation*}
  \begin{array}{rcl}
    \pr : \left(\matfqs\right)^m & \longrightarrow & \F_q^{m\ell} \\
    \left[
      \begin{pmatrix}
        a_{11}^1     & \dots & a_{1\ell}^1 \\
        \vdots       &       & \vdots \\
        a_{\ell 1}^1 & \dots & a_{\ell\ell}^1
      \end{pmatrix},\dots,
      \begin{pmatrix}
        a_{11}^m     & \dots & a_{1\ell}^m \\
        \vdots       &       & \vdots \\
        a_{\ell 1}^m & \dots & a_{\ell\ell}^m
      \end{pmatrix}
    \right]
    & \longmapsto &
    (a_{11}^1,\dots,a_{\ell 1}^1,\dots,a_{11}^m,\dots,a_{\ell 1}^m).
  \end{array}
\end{equation*}
\begin{algorithm}
\caption{Welch-Berlekamp for quasi-BCH codes}
\label{al:wb}
\begin{algorithmic}[1]
  \REQUIRE a received vector $y \in \F_q^{m\ell}$ with at most
           $\tau$ errors.
  \ENSURE the unique codeword within distance $\tau$ of $y$.
  \STATE $(Z_1,\dots,Z_m) \gets \psi(y)$
         where $\psi$ is the map from Theorem~\ref{thm:qbch_rs}.
  \STATE Find $Q = Q_0(X) + Q_1(X) Y \in (\matfqs[X])[Y]$ of
         degree $1$ such that
         \begin{enumerate}
           \item $(\Gamma^{i - 1},Z_i)Q = 0$
                 for all $i = 1,\dots,m - 1$,
           \item $\deg Q_0 \leq n - \tau - 1$,
           \item $\deg Q_1 \leq n - \tau - 1 - (k - 1)$.
         \end{enumerate}
  \STATE $f \gets$ the unique root of $Q$ in $\poly{(\matfqs)}{k}$
         such that
         $d \left( (Z_1,\dots,Z_m),
                   ((I_{\ell})f,\dots,(\Gamma^{m - 1})f)
            \right) \leq \tau$.
  \RETURN $\pr \left(
             (I_{\ell})f,(\Gamma)f,\ldots,(\Gamma^{m - 1})f
           \right)$.
\end{algorithmic}
\end{algorithm}
\begin{lemma}
\label{lem:wb_interpolation}
  Let $y \in \F_q^{m\ell}$ be a received word containing
  at most $\tau$ errors.
  Then there exists a nonzero bivariate polynomial
  $Q = Q_0 + Q_1 Y \in (\matfqs)[X,Y]$ satisfying
  \begin{enumerate}
    \item $(\Gamma^{i - 1},Z_i)Q = 0$ for $i = 1,\ldots,n$.
    \item $\deg Q_0 \leq n - \tau - 1$.
    \item $\deg Q_1 \leq n - \tau - 1 - (k - 1)$.
  \end{enumerate}
\end{lemma}
\begin{\myproof}
  We solve the problem with linear algebra over $\F_{q^s}$.
  We have, for each column of the solution, $n\ell$ equations and
  $\ell \left[ (n - \tau) + (n - \tau - (k - 1))\right]
   = \ell (n + 1)$ unknowns by Proposition~\ref{prop:rs}.
\end{\myproof}
\begin{lemma}
\label{lem:wb_roots}
  Let $Q \in (\matfqs)[X,Y]$ satisfying the three
  conditions of Lemma~\ref{lem:wb_interpolation} and
  $f \in \poly{(\matfqs)}{k}$ be such that
  $d((Z_1,\dots,Z_m),
     ((I_{\ell})f,\dots,(\Gamma^{m - 1})f)) \leq \tau$.
  Then $(X,f(X))Q = 0$.
\end{lemma}
\begin{\myproof}
  The polynomial $(X,f(X))Q$ has degree at most $n - \tau - 1$. By
  Lemma~\ref{lem:bivariate_degree1_ev} we have
  $(\Gamma^{i - 1})((X,f(X))Q) = (\Gamma^{i - 1},(\Gamma^{i - 1})f)Q
   = (\Gamma^{i - 1},Z_i)Q = 0$ for at least
  $n - \tau$ values of $i \in \{1,\ldots,n\}$. And therefore
  we must have $(X,f(X))Q = 0$.
\end{\myproof}
\begin{proposition}
\label{prop:wb_correct}
  Algorithm~\ref{al:wb} works correctly as expected and can correct
  up to $\floor{\frac{\delta - 1}{2}}$ errors.
\end{proposition}
\begin{proof}
  This is a direct consequence of Lemmas~\ref{lem:wb_interpolation}
  and~\ref{lem:wb_roots}.
\end{proof}

%% file: interleaved.tex
In this Section we prove that quasi BCH codes can be viewed as
an interleaving of classical BCH codes. We fix for this Section
$\Gamma \in \matfqs$ a primitive $m$-th root of unity and
$\C = \qbch_q(m,\ell,\delta,\Gamma)$. We first recall the
definition of interleaved codes.
\begin{definition}
  Let $\C_1,\dots,\C_{\ell}$ be error correcting codes over
  $\F_q$. The \emph{interleaved code $\C$ with respect to
  $\C_1,\dots,\C_{\ell}$} is a subset of $M_{\ell \times m}(\F_q)$,
  equipped with the $\ell$-bloc distance with respect to the columns,
  such that $c \in \C$ if and only if the $i$-th row of $c$ is
  a codeword of $\C_i$ for $i = 1,\dots,\ell$.
\end{definition}
\begin{lemma}
\label{lem:gamma_diag}
  The matrix $\Gamma$ diagonalizes over an extension of $\F_{q^s}$
  and its eigenvalues are all primitive $m$-th roots of unity.
\end{lemma}
\begin{proof}
  Let $\F_{q^{s'}} \supseteq \F_{q^s}$ be the splitting field of
  $X^m - 1$. The polynomial $X^m-1$ is a multiple of the minimal
  polynomial $\mu(X)$ of $\Gamma$. Hence the egeinvalues of $\Gamma$
  are $m$-roots of unity. Let $P \in \GL_{\ell}(\F_{q^{s'}})$ be such
  that $P^{-1} \Gamma P$ is diagonal. Now if an eigenvalue
  $\lambda_i$ of $\Gamma$ has order $d < m$, then
  \begin{equation*}
    P^{-1} (\Gamma^d - I_{\ell}) P =
    \begin{pmatrix}
      \lambda_1^d &        &             &        & & \\
                  & \ddots &             &        & & \\
                  &        & \lambda_i^d &        & & \\
                  &        &             & \ddots & & \\
                  &        &             &        & \lambda_\ell^{d}
    \end{pmatrix} - I_{\ell}
  \end{equation*}
  is singular as its $i$-th diagonal element would be zero.
  Consequently $\Gamma^d - I_{\ell} \not\in \GL_{\ell}(\F_{q^{s'}})$
  which is absurd.
\end{proof}
\begin{theorem}
\label{thm:interleaved}
  The quasi-BCH code $\C$ over $\F_q$ is an interleaved code
  of $\ell$ subcodes of Reed-Solomon codes over $\F_{q^{s'}}$
  in the following sense:
  there exists $\ell$ Reed-Solomon codes $\C_1,\dots,\C_{\ell}$ over
  $\F_q$ and an isometric isomorphism from $\C$, equipped with the
  $\ell$-block distance, to a subcode of the interleaved code with
  respect to $\C_1,\dots,\C_{\ell}$.
\end{theorem}
\begin{proof}
  We take the notation of the proof of Lemma~\ref{lem:gamma_diag}.
  Recall that
  \begin{equation*}
    H = \begin{pmatrix}
      I_{\ell} & \Gamma         & \cdots & \Gamma^{m - 1} \\
      I_{\ell} & \Gamma^2       & \cdots & \Gamma^{2(m - 1)} \\
      \vdots   & \vdots         &        & \vdots \\
      I_{\ell} & \Gamma^{\delta - 1} & \cdots 
                                       & \Gamma^{(\delta - 1)(m - 1)}
    \end{pmatrix} \in M_{(\delta - 1)\ell,m\ell}(\F_{q^s})
  \end{equation*}
  is a parity check matrix for $\C$
  (proof of Theorem~\ref{thm:qbch_rs}).
  By Lemma~\ref{lem:gamma_diag} we have that
  \begin{multline*}
    \left(
      c_{11} , \dots , c_{1 \ell} , \dots ,
      c_{m1} , \dots , c_{m \ell}
    \right) \in \C
    \Longleftrightarrow \\
    \begin{pmatrix}
      P^{-1} &        & \\
             & \ddots & \\
             &        & P^{-1}
    \end{pmatrix}
    \begin{pmatrix}
      I_{\ell} & \Gamma         & \cdots & \Gamma^{m - 1} \\
      I_{\ell} & \Gamma^2       & \cdots & \Gamma^{2(m - 1)} \\
      \vdots   & \vdots         &        & \vdots \\
      I_{\ell} & \Gamma^{\delta - 1} & \cdots 
                                       & \Gamma^{(\delta - 1)(m - 1)}
    \end{pmatrix}
    \begin{pmatrix}
      P &        & \\
        & \ddots & \\
        &        & P
    \end{pmatrix} \times \\
    \left[
    \begin{pmatrix}
      P^{-1} &        & \\
             & \ddots & \\
             &        & P^{-1}
    \end{pmatrix}
    \begin{pmatrix}
      c_{11} \\ \vdots \\ c_{1 \ell} \\ \vdots \\
      c_{m1} \\ \vdots \\ c_{m \ell}
    \end{pmatrix} \right] = 0 \\
    \text{and } 
    \left( c_{11} , \dots , c_{1 \ell} , \dots ,
      c_{m1} , \dots , c_{m \ell} \right) \in \F_q^{m\ell}
  \end{multline*}
  Let
  \begin{equation}
  \label{equ:iso}
    \begin{pmatrix}
      v_{11} \\ \vdots \\ v_{1 \ell} \\ \vdots \\
      v_{m1} \\ \vdots \\ v_{m \ell}
    \end{pmatrix}
    =
    \begin{pmatrix}
      P^{-1} &        & \\
             & \ddots & \\
             &        & P^{-1}
    \end{pmatrix}
    \begin{pmatrix}
      c_{11} \\ \vdots \\ c_{1 \ell} \\ \vdots \\
      c_{m1} \\ \vdots \\ c_{m \ell}
    \end{pmatrix}
  \end{equation}
  Denote by $\sigma$ the application defined by
  equation~\eqref{equ:iso}. Then
  \begin{multline}
  \label{equ:bch}
    \left(
      c_{11} , \dots , c_{1 \ell} , \dots ,
      c_{m1} , \dots , c_{m \ell}
    \right) \in \C \Longleftrightarrow \\
    \sigma^{-1} \left(
      v_{11} , \dots , v_{1 \ell} , \dots ,
      v_{m1} , \dots , v_{m \ell}
    \right) \in \F_q^{m\ell}
    \text{ and for } i = 1,\dots,\ell \\
    \begin{pmatrix}
      1 & \lambda_i              & \dots & \lambda_i^{m - 1} \\
      1 & \lambda_i^2            & \dots & \lambda_i^{2(m - 1)} \\
      \hdotsfor{4} \\
      1 & \lambda_i^{\delta - 1} & \dots &
                                \lambda_i^{(\delta - 1)(m - 1)}
    \end{pmatrix}
    \begin{pmatrix}
      v_{1i} \\ \vdots \\ v_{mi}
    \end{pmatrix} = 0.
  \end{multline}
  Then it is straightforward that
  $\sigma$ is an isometric isomorphism from $\C$ equipped with
  the $\ell$-block distance and $\sigma(\C)$, which is by
  equation~\eqref{equ:bch} a subcode of
  the interleaved code with respect
  to $\ell$ subcodes of Reed-Solomon codes over $\F_q$.
  For $i = 1,\dots,\ell$ take $\C_i$ to be
  the Reed-Solomon code defined by the parity check matrix
  of equation~\eqref{equ:bch}.
\end{proof}

Note that if the minimal polynomial of $\Gamma$ has degree one:
$\Gamma = X - \lambda$, then $s' = s$ and $\Gamma$ diagonalizes as
$\lambda I_{\ell}$. Consequently the Reed-Solomon codes
$\C_1,\ldots, \C_{\ell}$ are isomorphic, as they are defined by
the same control equations in equation~\eqref{equ:bch}.
In such a case, we can apply the result on the correction capacity
for interleaved Reed-Solomon codes~\cite{SchSiBo2006,BleKiaYun2007}.



%
\begin{corollary}
  There exists a decoding algorithm that is guaranteed to correct up to
  $\frac{\delta - 1}{2}$ errors. In particular, if
  the minimal polynomial of $\Gamma$ has degree~$1$
    over $\F_{q^{s}}$ 
then it can correct up to $\frac{\ell}{\ell + 1}(\delta - 1)$
  errors with high probability.
\end{corollary}
\begin{proof}
  Taking the notation of Theorem~\ref{thm:interleaved} and
  if $y = c + e$ is a received word, one can decode $\sigma(y)$
  with the decoding algorithms of $\C_1,\dots,\C_{\ell}$ obtaining
  $c' \in \F_{q^{s'}}^{m\ell}$. Then $c = \sigma^{-1}(c')$.

  If the minimal polynomial of $\Gamma$ has degree~$1$, then
  $\C_1 = \C_2 = \dots = \C_{\ell}$ and one can apply the
  algorithm of~\cite{BleKiaYun2007} or \cite{SchSiBo2006}.
\end{proof}